\documentclass[12pt]{article}
\usepackage{amssymb,amsmath}
\usepackage{amsthm}
\usepackage{xcolor}
\usepackage{url}
\usepackage[backref]{hyperref}
\newtheorem{theorem}{Theorem}
\newtheorem{proposition}[theorem]{Proposition}

\newtheorem{conjecture}[theorem]{Conjecture}
\theoremstyle{definition}
\newtheorem{definition}[theorem]{Definition}

\newtheorem{remark}[theorem]{Remark}
\newcommand{\beq}{\begin{equation}}
\newcommand{\eeq}{\end{equation}}

\newcommand{\p}{\partial}

\allowdisplaybreaks[3]

\newcommand*{\pd}
[2]{\mathchoice{\frac{\partial#1}{\partial#2}}
  {\partial#1/\partial#2}{\partial#1/\partial#2}
  {\partial#1/\partial#2}}

\newcommand*{\fd}
[2]{\mathchoice{\frac{\delta#1}{\delta#2}}
  {\delta #1/\delta#2}{\delta#1/\delta#2}{\delta#1/\delta#2}}

\begin{document}

\title{Bi-Hamiltonian structures of WDVV-type}
\author{S. Opanasenko ${}^{*,^\ddag}$ and R. Vitolo ${}^{*,**}$}
\date{}
\maketitle
\vspace{-7mm}
\begin{center}
  $^{**}$Dipartimento di Matematica e Fisica ``E. De Giorgi'',\\
  Universit\`a del Salento\\
  and $^{*}$Sezione INFN di Lecce\\
  via per Arnesano, 73100 Lecce, Italy
  \\
  $^\ddag$\,Institute of Mathematics of NAS of Ukra\"ine,\\
  3 Tereshchenkivska Str., 01024 Kyiv, Ukra\"ine\\
  e-mails:
  \texttt{stanislav.opanasenko@le.infn.it}\\
  \texttt{raffaele.vitolo@unisalento.it}
\end{center}

\begin{abstract}
  We study a class of nonlinear PDEs that admit the same bi-Hamiltonian
  structure as WDVV equations: a Ferapontov-type first-order Hamiltonian
  operator and a homogeneous third-order Hamiltonian operator in a canonical
  Doyle--Pot\"emin form, which are compatible.  Using various equivalence
  groups, we classify such equations in two-component and three-component
  cases.  In a four-component case, we add further evidence to the conjecture
  that there exists only one integrable system of the above type.  Finally, we
  give an example of the six-component system with required bi-Hamiltonian
  structure.  To streamline the symbolic computation we develop an algorithm to
  find the aforementioned Hamiltonian operators, which includes putting forward
  a conjecture on the structure of the metric parameterising the first-order
  Hamiltonian operator.
\end{abstract}

\tableofcontents

\section{Introduction}

\subsection{Integrable Systems and WDVV equations}

In the area of infinite-dimensional integrable systems, bi-Hamiltonian systems
of partial differential equations (PDEs) play a central role
\cite{KrasilshchikVinogradov:SCLDEqMP,Magri:SMInHEq,Magri:SInHP,
  NovikovManakovPitaevskiiZakharov:TS,Olver:ApLGDEq,Zakharov:WIsIn}.
The presence of two independent compatible Hamiltonian structures ensures,
under certain mild hypotheses~\cite{sole14:_lenar}, the existence of infinite
sequences of commuting conserved quantities and commuting symmetries, thus
mimicking the Liouville integrability of the finite-dimensional case.

It was proved by B. Dubrovin that the solutions of
Witten--Dijkgraaf--Verlinde--Verlinde (WDVV) equations can be put in
correspondence with a large class of bi-Hamiltonian structures
(see~\cite{D96}), thus providing a bridge between Topological Field Theories
and Integrable Systems. The consequences in mathematics have been far-reaching,
as solutions of WDVV equations yield Dubrovin--Frobenius manifolds, which
nowadays is an active research topic.

In dimension~$N$, $F=F(t^1,\dots,t^N)$, WDVV equations is a nonlinear
overdetermined system of PDEs,
\[
  \eta^{\lambda\mu}F_{\lambda\alpha\beta}F_{\mu\nu\gamma} =
  \eta^{\lambda\mu}F_{\lambda\alpha\nu}F_{\mu\beta\gamma},
  \quad\text{where}\quad F_{\alpha\beta\gamma}:=\frac{\p^3F}{\p t^\alpha\p
    t^\beta \p t^\gamma}.
\]
The inverse~$(\eta_{\alpha\beta})$ of the constant symmetric nondegenerate
matrix~$(\eta^{\alpha\beta})$ is by definition
$\eta_{\alpha\beta}:=F_{1\alpha\beta}$.  The WDVV equations are equivalent to
the requirement of associativity of a product operation with structure
constants $c^\alpha_{\beta\gamma} =\eta^{\alpha\nu} F_{\nu\beta\gamma}$, and
therefore they are sometimes also called \emph{associativity equations}.  Note
that the requirements on~$F$ completely specify its dependence on~$t^1$ up to
second degree polynomials,
\begin{equation*}
  F=\frac{1}{6}\eta_{11}(t^1)^3 +
  \frac{1}{2}\sum_{k>1}\eta_{1k}t^k(t^1)^2 + \frac{1}{2}\sum_{k,s>1}
  \eta_{sk}t^st^kt^1 + f(t^2,\ldots,t^N).
\end{equation*}

The fact that systems whose solutions yield integrable systems can themselves
be formulated as integrable systems is a common phenomenon in the
field of integrable systems.  Thus, it was proved~\cite{FGMN97} that a system
of WDVV equations, after being reformulated as a quasilinear first-order system
of PDEs in conservative form,
\begin{equation}
  \label{eq:ConsHydroSys}
  u^i_t = (V^i)_x,
\end{equation}
admits a bi-Hamiltonian formulation:
\begin{equation}\label{eq:6}
  u^i_t = A^{ij}_k\fd{H_k}{u^j}, \qquad k=1,2
\end{equation}
(we use Einstein's summation convention throughout the paper).  Here
$u=(u^1,\dots,u^n)$ is the vector of field variables, $(t,x)$ are the
independent variables and $V^i$'s are smooth functions of the dependent
variables.  Apart from a conservative form, the quasilinear system we consider
have another two important properties: non-diagonalisability and linear
degeneracy.  Recall that a quasilinear system $u^i_t=V^i_ju^j_x$ is called
diagonalisable if there exists a change of coordinates $r^i=r^i(u)$ such that
the transformed system is diagonal, $r^i_t=\tilde V^ir^i_x$, and linearly
degenerate if $L_p v = 0$ for every pair $(v,p)$, where~$p$ is the right
eigenvector of the matrix~$(V^i_j)$ corresponding to its eigenvalue~$v$ ($L$ is
the Lie derivative).  The pair of compatible Hamiltonian operators $A_1$, $A_2$
(compatibility means $A_1+\lambda A_2$ is a Hamiltonian operator for any
$\lambda\in \mathbb{R}$ as well) was \emph{different} from the pair of
compatible Hamiltonian operators which is determined by a solution of WDVV
equations.  Bi-Hamiltonianity was further confirmed for other simple cases of
WDVV equations in~\cite{kalayci97:_bi_hamil_wdvv,kalayci98:_alter_hamil_wdvv}.
After some years, advances both in the theory \cite{FPV17:_system_cl} and in
Symbolic Computations~\cite{vasicek22:_wdvv,vasicek21:_wdvv_hamil,
  vitolo17:_hamil_pdes,KVV17} led to the discovery of new bi-Hamiltonian
structures for more complicated WDVV
equations~\cite{PV15,m.v.19:_bi_hamil_orien_assoc,vasicek21:_wdvv_hamil,
  vasicek22:_wdvv}.

\subsection{Bi-Hamiltonian structures for WDVV systems}

A matrix differential operator $A^{ij}=a^{ij\sigma}(u_\tau)D_\sigma$, where
$D_\sigma=D_x^{\sigma}$, $u_\tau$ stands for a finite collection of derivatives
of the dependent variables and $\sigma$, $\tau\in\mathbb{N}_0$, is called
Hamiltonian if the corresponding bracket
\begin{equation*}
  \{F,G\}_A = \int\fd{F}{u^i}A^{ij}\fd{G}{u^j}\,\mathrm dx
\end{equation*}
is Poisson. More precisely, skew-symmetry of the bracket is equivalent to
skew-adjointness of~$A$, and the Jacobi property is equivalent to differential
conditions on the coefficients of~$A$ that have an intrinsic formulation, the
vanishing of the Schouten bracket: $[A,A]=0$. In turn, compatibility of two
Hamiltonian operators $A_1$, $A_2$ is equivalent to the vanishing of the
Schouten bracket, $[A_1,A_2]=0$.

Both Hamiltonian operators $A_1$, $A_2$ admitted by a WDVV system are
homogeneous, where homogeneity is defined with respect to the grading
$\deg{D_x}=1$.  Such operators were introduced
in~\cite{DubrovinNovikov:PBHT,DN83} as a family of operators that is invariant
with respect of diffeomorphisms of the space of field variables.

In particular, the operator~$A_1$ a first-order homogeneous Hamiltonian
operator~$A_1=P$ of Ferapontov type,
\begin{equation}
  \label{eq:16}
  P^{ij}=g^{ij}\mathrm D_x + \Gamma^{ij}_k u^k_x
  + c^{\alpha\beta}w^i_{\alpha h}u^h_x\mathrm D_x^{-1}w^j_{\beta k}u^k_x,
\end{equation}
and an operator~$A_2$ is the third-order Hamiltonian operator~$A_2=R$,
\begin{equation}\label{eq:18}
R^{ij}=\mathrm D_x(f^{ij}\mathrm D_x + c^{ij}_k u^k_x)\mathrm D_x,
\end{equation}
which is compatible with~$P$. All coefficient functions~$g^{ij}$, $f^{ij}$,
$\Gamma^{ij}_k$, $c^{ij}_k$ and~$w^i_{\alpha h}$ above are functions of the
field variables~$u$ only, and $c^{\alpha\beta}$ are constants.

For example, in the simplest case, $N=3$ and
\begin{displaymath}
  \eta=
  \begin{pmatrix}
    0 & 0 & 1\\ 0 & 1 & 0\\ 1 & 0 & 0
  \end{pmatrix}
\end{displaymath}
the WDVV equations reduce to the form
\begin{equation}\label{eq:26}
  f_{ttt} = f_{txx}^2 - f_{xxx}f_{ttx}.
\end{equation}
Following~\cite{FM96:_equat_hamil},
we introduce new dependent variables $u^1=f_{xxx}$, ${u^2=f_{txx}}$,
$u^3=f_{ttx}$. Then, the equation~\eqref{eq:26} can be rewritten in the form
\eqref{eq:ConsHydroSys} as
\begin{equation}
  \label{eq:27}
u^1_t = u^2_x,
    \quad
u^2_t = u^3_x,
    \quad
u^3_t = ((u^2)^2 - u^1u^3)_x.
\end{equation}
It was found in~\cite{FGMN97} that the system~\eqref{eq:27} is bi-Hamiltonian
with the local (\emph{i.e.} $c^{\alpha\beta}=0$) first-order homogeneous
operator~$P$ and the third-order homogeneous operator~$R$,
\begin{gather*}
  P=
  \begin{pmatrix}
    -\frac{3}{2}\mathrm D_x & \frac{1}{2}\mathrm D_x u^1 & \mathrm D_x u^2
    \\
    \frac{1}{2}u^1\mathrm D_x & \frac{1}{2}(\mathrm D_x u^2 + u^2\mathrm D_x) &
    \frac{3}{2}u^3\mathrm D_x + u^3_x
    \\
    u^2\mathrm D_x &  \frac{3}{2}\mathrm D_xu^3 - u^3_x &
    ((u^2)^2-u^1u^3)\mathrm D_x + \mathrm D_x((u^2)^2-u^1u^3)
  \end{pmatrix},\\[1ex]
  R = \mathrm D_x
    \begin{pmatrix}
      0 & 0 & \mathrm D_x \\
      0 & \mathrm D_x & -\mathrm D_xu^1 \\
      \mathrm D_x & -u^1\mathrm D_x &
      \mathrm D_x u^2 + u^2\mathrm D_x + u^1 \mathrm D_x u^1
    \end{pmatrix}\mathrm D_x,
\end{gather*}
The corresponding Hamiltonian densities are
\[
  h_P=u^3\quad\text{and}\quad h_R= - \frac{1}{2}u^1\left(\mathrm
    D_x^{-1}u^2\right)^2 - (\mathrm D_x^{-1}u^2) (\mathrm D_x^{-1}u^3).
\]
\begin{remark}
  We stress that the type of the bi-Hamiltonian pair for WDVV equations is
  different from the bi-Hamiltonian pair determined by a solution of WDVV
  equations. Indeed, the latter is constituted by two compatible local
  homogeneous Hamiltonian operators of order one
  \cite{dubrovin98:_flat_froben,D96}.
\end{remark}

\subsection{Problem and results}

\begin{definition}
  We define a \emph{bi-Hamiltonian structure of WDVV-type} to be a pair of
  compatible Hamiltonian operators $P$, $R$ as in \eqref{eq:16}, \eqref{eq:18},
  respectively.

  We say a quasilinear first-order system of PDEs in conservative
  form~\eqref{eq:ConsHydroSys} to be a \emph{bi-Hamiltonian system of WDVV-type} if it is
  endowed with a bi-Hamiltonian structure of WDVV-type.
\end{definition}

\textbf{The problem}: \emph{In this paper, we aim at classifying (when
  possible) bi-Hamiltonian equations of WDVV-type, and at the same time at
  introducing new bi-Hamiltonian equations of WDVV-type, which are not
  necessarily related with WDVV equations.}

As this task presents several theoretical and computational
challenges, our paper contains both the theoretical advances and the
computational algorithms that made this research possible.

\textbf{The main results} of the paper are listed below, following the number of
unknown functions~$u^i$. Note that we discard linear bi-Hamiltonian systems of
WDVV type as we regard them as trivial.
\begin{description}
\item[$n=2$] We have an affine classification of bi-Hamiltonian structures of
  WDVV-type, according to which there exist two classes of nonlinear
  bi-Hamil\-to\-nian systems of WDVV type. Both cases turn out to be
  linearisable if we enlarge the group action to the group of projective
  reciprocal transformations that preserve~$t$.
\item[$n=3$] In this case, under the action of the group of projective
  reciprocal transformations that preserve~$t$, we have five nontrivial
  cases. Three of them are particular types of WDVV equations known in the
  literature, while two of them are new. In order to compute the leading
  coefficient of the operator $P$ we put forward Conjecture~\ref{conj:comp-g}
  in Subsection~\ref{sec:conj-struct-g_1} that proved to be true in all known
  cases.

  Then, under the action of the full group of projective reciprocal
  transformations, it can be shown that all nontrivial cases reduce to one, the
  simplest WDVV equation. Unfortunately, the corresponding transformation is
  practically impossible to find with current computational tools. This
  implies that our methods, which allow to find the bi-Hamiltonian structure,
  are still highly valuable.
\item[$n=4$] In this case, systems with third-order Hamiltonian structures have
  been classified, however, only one of them is known to be integrable. The
  system is obtained within the geometric theory of linearly degenerate systems
  in the Temple class \cite{agafonov05:_integ} (no relationship with WDVV).  We
  found a new bi-Hamiltonian structure of WDVV-type for this system.  In
  some other interesting cases we proved that such structures do not exist,
  although we cannot guarantee it in general.
\item[$n=6$] We proved that two commuting systems of the type~\eqref{eq:ConsHydroSys} are
  bi-Hamil\-to\-nian systems of WDVV-type. These are two integrable systems
  that arise from integrability conditions of a certain class of
  Lagrangians~\cite{FPX2021}. The systems are not related with WDVV equations;
  this, together with the previous item, proves that the class of
  bi-Hamiltonian systems of WDVV-type does not reduce to WDVV equations only.
\end{description}

It is worth to compare bi-Hamiltonian structures in this paper with another
family of bi-Hamiltonian structures that has been introduced
in~\cite{LSV:bi_hamil_kdv}.  Such structures are of the form $A_1=P$,
$A_2=Q+\epsilon^2 R$, and are said to be bi-Hamiltonian structures of
KdV-type. The operators~$(P,Q,R)$ are a compatible triple of homogeneous
Hamiltonian operators, with~$P$, $Q$ being of the form~\eqref{eq:16} and $R$
being a higher order homogeneous Hamiltonian operator. Note that $\epsilon$ is
introduced as a perturbative parameter. For
example, the triple
\begin{equation*}
  P=D_x,\quad Q=\frac{2}{3}u\mathrm D_x + \frac{1}{3}u_x,
\quad R=\mathrm D_x^3
\end{equation*}
provides a bi-Hamiltonian structure for the KdV equation. Many integrable
systems have this form (see~\cite{LSV:bi_hamil_kdv} for a detailed list of
examples and references).

\begin{remark}
When comparing bi-Hamiltonian structures of KdV and WDVV type, it is evident
that \emph{we can treat the bi-Hamiltonian structures of WDVV type as a
  singular case of bi-Hamiltonian structures of KdV type}. The `singularity' of
the bi-Hamiltonian pair goes in opposite direction with respect to the usual
dispersionless limit $\epsilon\to 0$.
\end{remark}

The results that we obtained in this paper show that the class of
bi-Hamiltonian structures of WDVV type is rich and interesting, and it deserves
further investigation.

\subsection{Computational problems and their solutions}
\label{sec:comp-probl-their}

Generally speaking, finding Hamiltonian operators for systems of PDEs of the
type \eqref{eq:ConsHydroSys} in high ($n\geq 3$) dimension and proving their
compatibility is not an easy task. The computational problems to be solved are
summarised below.
\begin{itemize}
\item Finding a Hamiltonian formulation \eqref{eq:6} through a first-order
  operator~$P$ as in~\eqref{eq:16} amounts to solving a complicated system of
  nonlinear PDEs. Such a system is not easily solvable even when $n=3$.
\item On the other hand, third-order operators are
  classified~\cite{FPV14,FPV16} in low dimensions, as well as the systems of
  the form~\eqref{eq:ConsHydroSys} that admit a corresponding Hamiltonian
  formulation~\cite{FPV17:_system_cl}. Hence we can start from a third-order
  operator $R$ from the classification and the corresponding systems, and find
  first-order Hamiltonian formulations~\eqref{eq:6}.
\item First-order operators $P$ for the above systems of PDEs are found by
  solving the system mentioned above, with the additional requirement of
  compatibility: $[P,R]=0$.
\end{itemize}

Compatibility is the major computational problem here. First of all, until
recently, there did not exist a way to bring Schouten bracket between nonlocal
operators in a canonical form, and require its vanishing. In \cite{CLV19} an
algorithm was presented. However, finding a \emph{minimal} set of conditions
which are equivalent to $[P,R]=0$ is the problem that is still unresolved.

In~\cite{m.20:_weakl_poiss} the above algorithm was implemented in the computer
algebra systems Maple, Mathematica and Reduce. Thus, in principle, we can
compute $[P,R]$ and require its vanishing if one of the two operators is
unknown. However, the computational complexity is sometimes overwhelming, even
for dedicated compute servers, and the main problem is to keep the complexity
within the limit of what our computers can do.

The requirement $[P,R]=0$ for unknown $P$ leads to a complicated overdetermined
nonlinear system of PDEs that we are able to solve only in dimension $n=2$ and
in the simplest case in dimension $3$. Working in Maple computer algebra
system, we used \texttt{rifsimp} to reduce the equations to an involutive
system with a known dimension of finite-dimensional solution space and then
\texttt{pdsolve} to solve it.

For more general calculations, we used Conjecture~\ref{conj:comp-g}; it yields
a condition that we always verified in concrete compatible operators~$P$
and~$R$.  Such a condition is very likely to be an important part of the
compatibility conditions, and brings again the problem of finding $P$ to a
system of algebraic equations, which becomes manageable also in high
dimensions.

However, that is not all: we still need to check that $[P,R]=0$, although at
this point that is only a straightforward computation as $P$ and $R$ are
completely specified. Not all such calculations can be performed on a modern
laptop: primarily, huge amounts of RAM are needed as the formula of the
Schouten bracket (see Subsection~\ref{sec:conj-struct-g_1}) and the algorithm
for bringing the expressions to a canonical form~\cite{CLV19} imply the
calculation of iterated derivatives of large rational expressions, leading to
expression swell that only at the very end simplifies to zero (if that is the case).

In particular, the largest successful direct calculations have been performed
for the calculations of $[P,Q]$ in the case $R^{(2)}$ ($n=3$) ($15$GB of RAM
and $90$ hours of computing time) and in the case $n=4$.  In the case $R^{(1)}$
($n=3$) the calculation of $[P,Q]$ failed for lack of RAM.  This led to a
re-analysis of the algorithm in~\cite{CLV19} and a work-around has been found.
Since in the computation all denominators are of the form $\sigma^n$,
where~$\sigma$ is the singular variety of the corresponding Monge metric,
while~$n$ varies, an expression~$\sigma^{-1}$ was given a new notation,
effectively making all rational expressions polynomials.  This modification led
to a much reduced RAM consumption in case~$R^{(2)}$ (only $2$GB!), and manyfold
reduction of computing time, and allowed us to prove compatibility of the
Hamiltonian operators~$R^{(1)}$ and~$P^{(1)}$ using $45$GB of RAM and $54$
hours of computing time.

Surprisingly, the case $n=6$ turned out to be manageable on a laptop with
$16$GB of RAM, mostly due to its locality.

All large calculations have been performed on a compute server of the INFN,
Section of Lecce. The server has a processor AMD EPYC 7282 and 256GB of RAM.
The most important computations that have been programmed by us
are available at \cite{wdvv-sw}; we will be happy to help other researchers who
are interested in our software or similar computational tasks.

\section{A pair of Hamiltonian structures}

Our goal is to classify quasilinear first-order systems of evolutionary PDEs of
the type~\eqref{eq:ConsHydroSys} that are endowed with a pair of Hamiltonian
operators~\eqref{eq:16} and~\eqref{eq:18}.

\subsection[First-order Hamiltonian operator of Ferapontov type]{First-order Hamiltonian operator\\ of Ferapontov type}
\label{sec:first-order-hamilt}

A Hamiltonian operator~$P$,
\[
  P^{ij}=g^{ij}\mathrm D_x + \Gamma^{ij}_s u^s_x
  + c^{\alpha\beta}w^i_{\alpha s}u^s_x\mathrm D_x^{-1}w^j_{\beta t}u^t_x,
\]
is a nonlocal generalisation of a classical Dubrovin--Novikov Hamiltonian operator ($c^{\alpha\beta}=0$),
which was introduced by Ferapontov in~\cite{F95:_nl_ho}.
The operator~$P$ is Hamiltonian if and only if the following properties are fulfilled
(we assume the non-degeneracy condition $\det (g^{ij})\neq 0$):
\begin{subequations}\label{eq:17}
  \begin{gather}
    g^{ij}=g^{ji},
    \quad
    g^{ij}_{,k} = \Gamma^{ij}_k + \Gamma^{ji}_k,
    \quad
    g^{is}\Gamma^{jk}_s = g^{js}\Gamma^{ik}_s,
    \\
    g^{is}w^j_{\alpha s} = g^{js}w^i_{\alpha s},\label{eq:185}
    \\
    \nabla_k w^i_{\alpha j} = \nabla_j w^i_{\alpha k},\label{eq:187}
    \\
    [w_\alpha,w_\beta] = 0,\label{eq:19}
    \\
    c^{\alpha\beta}=c^{\beta\alpha},\label{eq:19.5}
    \\
    R^{ij}_{kl} = c^{\alpha\beta}\Big( w^i_{\alpha l}w^j_{\beta k} - w^i_{\alpha k}w^j_{\beta l}\Big).\label{eq:20}
  \end{gather}
\end{subequations}
Here and below, $g^{ij}_{,k}=\pd{g^{ij}}{u^k}$ and
$w_\alpha=w^i_{\alpha j}u^j_x\pd{}{u^i}$. We observe that the
conditions~\eqref{eq:19.5} and \eqref{eq:20} differ from those found in the
literature, since the operator~$P$ is a slightly generalised form of an
operator originally introduced in~\cite{F95:_nl_ho}, where the
matrix~$(c^{\alpha\beta}$) is diagonal. Nonetheless, they can easily be
deduced from the calculations in~\cite{CLV19} after replacing the ansatz for
the first-order operator with the ansatz~\eqref{eq:16}.

The number~$m$ of the commuting vector fields (and hence the dimen\-si\-on~$m$
of the symmetric matrix $(c^{\alpha\beta})$ is not known \emph{a
  priori}. However, there is an important computational feature~\footnote{ It
  is not a theorem, but an experimental observation communicated to us by
  E.V. Ferapontov.}  of linearly degenerate non-diagonalisable systems is that
in low dimension ($n\leq 5$) they only admit two commuting flows of first
order, namely the $x$- and $t$-translations, represented by
\begin{equation*}
      w_1=V^i_ju^j_x\pd{}{u^i}\quad\text{and}\quad w_2=u^i_x\pd{}{u^i}.
\end{equation*}
All but one of our examples fall within this category of equations,
and the one that does not is local, see Section~\ref{sec:AndBeyound},
and therefore we shall restrict ourselves to this case hereafter.

The above operator $P$ is clearly invariant with respect to local
diffeomorphisms of the dependent variables: $\tilde{u}=\tilde{U}(u)$. Under
such transformations, $g^{ij}$ transform as a contravariant $2$-tensor, and
$\Gamma^{ij}_k$ transform as the contravariant Christoffel symbols of a linear
connection.  Thus, geometrically the conditions~\eqref{eq:17} mean that if
$(g_{ij})=(g^{ij})^{-1}$, then $\Gamma^i_{jk}=-g_{jh}\Gamma^{hi}_k$ are the
Christoffel symbols of the Levi-Civita connection of the metric~$g_{ij}$.  Its
curvature, expressed via the Kobayashi--Nomizu convention:
\begin{equation*}
  R^i_{jkl} = \Gamma^i_{lj,k} - \Gamma^i_{kj,l} + \Gamma^i_{kr}\Gamma^r_{lj}
    - \Gamma^i_{lr}\Gamma^r_{kj},
\end{equation*}
admits the expansion~\eqref{eq:20}, where
\begin{equation*}
  R^{ijk}_l=  g^{is}g^{jt}R^k_{tsl} =
  g^{is}(\partial_l\Gamma^{jk}_{s}-\partial_s\Gamma^{jk}_{l})
  +\Gamma^{ij}_s\Gamma^{sk}_{l} -\Gamma^{sj}_l\Gamma^{ik}_{s}
\end{equation*}
and $ R^{jk}_{hl}= g_{hi}R^{ijk}_l$.  The
equations~\eqref{eq:185}--\eqref{eq:20} are nothing else but the
Gauss--Peterson--Codazzi equations for submanifolds~$M^n$ with a flat normal
connection in a (pseudo-)Euclidean space of dimension~$n+N$, where the
metric~$g$ plays the role of the first quadratic form of~$M^n$, and $w_\alpha$
are the Weingarten operators corresponding to the field of pairwise orthogonal
unit normals.

It was also proved \cite{pavlov95:_conser_hamil} that the nonlocality type of~$P$ is preserved under linear transformations
of the independent variables~$(t,x)$.

Finally, see~\cite{F95:_nl_ho,MalNov2001} for a discussion on
how to define a Hamiltonian for the operator~$P$ in the nonlocal case.

\subsection[Third-order Hamiltonian operator \\
in the Doyle--Pot\"emin canonical form]
{Third-order Hamiltonian operator\\
  in the Doyle--Pot\"emin canonical form}
\label{sec:third-order-hamilt}

The most general form of a third-order homogeneous Hamiltonian operator,
in accordance with the definition given in \cite{DubrovinNovikov:PBHT}, is
\begin{equation}
  \label{eq:25}
    \begin{split}
    R^{ij}=
    & g_3^{ij}\mathrm D_{x}^{3}
      +b_{3\,s}^{ij}u_{x}^{s}\mathrm D_{x}^{2}
  +[c_{3\,s}^{ij}u_{xx}^{s}
  +c_{3\,st}^{ij}u_{x}^{s}u_ {x}^{t}]\mathrm D_{x}
  \\
  & +d_{3\,s}^{ij}u_{xxx}^{s}
  +d_{3\,st}^{ij}u_{x}^{s}u_{xx}^{t}
  +d_{3\,srt}^{ij}u_{x}^{s}u_{x}^{r}u_{x}^{t}.
\end{split}
\end{equation}
Nevertheless, in view of the computational difficulties, a minimal set of
conditions on the coefficient of the above operator $R$ that are equivalent to
its skew-adjointness and $[R,R]=0$ is not known.

It was independently proved in~\cite{doyle93:_differ_poiss} and~\cite{potemin91:PhDt, potemin97:_poiss}
that there always exists a change of
dependent variables that brings a Hamiltonian operator of the
form~\eqref{eq:25} into the \emph{Doyle--Pot\"emin canonical form}~\eqref{eq:18}.
This form has a drastically reduced set of coefficients, and the Hamiltonian
property of such an operator is equivalent to
\begin{subequations}\label{eq:505}
  \begin{gather}
    c_{ijk}=\frac{1}{3}(f_{ik,j}-f_{ij,k}), \label{eq:506}\\
    f_{ij,k}+f_{jk,i}+f_{ki,j}=0,\label{eq:507}\\
    c_{ijk,l}=-f^{pq}c_{pil}c_{qjk}.\label{eq:733}
  \end{gather}
\end{subequations}
Here, $c_{ijk}=f_{iq}f_{jp}c_{k}^{pq}$ and, of course, $(f_{ij})=(f^{ij})^{-1}$.

A covariant $2$-tensor $f$ satisfying equations \eqref{eq:507} must be a Monge metric of a quadratic line complex (see~\cite{FPV14} and reference therein).
This is an algebraic variety in the Pl\"ucker embedding of the projective space whose affine chart is $(u^i)$.
In particular, being Monge, $f_{ij}$ is a quadratic polynomial in the field variables~$u$.

It was further proved in~\cite{balandin01:_poiss,doyle93:_differ_poiss,potemin91:PhDt,potemin97:_poiss}
that the tensor $f_{ij}$ can be factorised as
\begin{equation}
  \label{eq:12}
  f_{ij} = \phi_{\alpha\beta}\psi^\alpha_i\psi^\beta_j,\quad \text{\Big(or, in a matrix form,} \quad f=\Psi\Phi\Psi^\top\Big)
\end{equation}
where $\phi$ is a constant non-degenerate symmetric matrix of dimension~$n$,
and
\begin{equation*}
  \psi _{k}^{\gamma }=\psi _{ks}^{\gamma }u^s+\omega _{k}^{\gamma } 
\end{equation*}
is a non-degenerate square matrix of dimension~$n$, with the
constants~$\psi _{ij}^{\gamma }$ and~$\omega _{k}^{\gamma }$ satisfying the
relations
\begin{gather*}
  \psi _{ij}^{\gamma }=-\psi _{ji}^{\gamma },
  \\
  \phi _{\beta \gamma }(\psi _{il}^{\beta }\psi _{jk}^{\gamma }+\psi_{jl}^{\beta }\psi _{ki}^{\gamma }
  +\psi _{kl}^{\beta }\psi _{ij}^{\gamma})=0,
  \\
  \phi _{\beta \gamma }(\omega _{i}^{\beta }\psi _{jk}^{\gamma
  }+\omega _{j}^{\beta }\psi _{ki}^{\gamma }+\omega _{k}^{\beta }\psi
  _{ij}^{\gamma })=0.
\end{gather*}
The interested reader can find the general expression of the Hamiltonian for
the operator~$R$ (provided it exists for a given system) in~\cite{FPV17:_system_cl}.

The system that admits the operator~$R$ in a canonical form is evidently conservative.
On the other hand, if a system admits such an operator, then different conservative forms thereof
may not necessarily admit third-order Hamiltonian operator in a canonical form,
see discussion in~\cite[p.~664]{FGMN97}.

 \subsection{Compatibility and a conjecture}
\label{sec:conj-struct-g_1}

The only missing part of the picture are conditions under which the
operators~$P$ and~$R$ are compatible, \emph{i.e.} their Schouten bracket
vanishes: $[P,R]=0$. We use the definition of Schouten bracket in \cite{CLV19}:
\begin{align*}
  [A_1,A_2](\psi^1,\psi^2,\psi^3) =
  &\big[(\ell_{A_1,\psi^1}(A_2(\psi^2)))\psi^3 + \text{cyclic}(\psi^1,\psi^2,\psi^3)
  \\
  &+(\ell_{A_2,\psi^1}(A_1(\psi^2)))\psi^3  +
      \text{cyclic}(\psi^1,\psi^2,\psi^3)\big],
  \notag
\end{align*}
where $\psi^a=(\psi^a_i(x,u_\sigma))$, $a=1,2,3$, are covectors and the square
brackets on the right-hand side mean that the result is to be considered up to
the image of $\mathrm D_x$ (\emph{i.e.}, `total divergencies').  In the above
formula we use the linearization of an operator $A$
\cite{KrasilshchikVinogradov:SCLDEqMP}: if, in local coordinates, one has
$A(\psi)^i= a^{ij\sigma}\mathrm D_\sigma\psi_j$ then
\begin{gather*}
  \ell_{A,\psi}(\phi)^i = \frac{\partial a^{it\sigma}}{\partial u^s_\tau}
  \mathrm D_\sigma\psi_t\mathrm D_\tau\phi^s,\quad
  \phi=(\phi^s(x,u_\sigma)).
\end{gather*}
In~\cite{CLV19}, presented is an algorithm to compute a divergence-free form of the bracket~$[A_1,A_2]$
for a wide class of nonlocal operators (\emph{weakly nonlocal operators}), to which Ferapontov-type operators belong.
Based on this algorithm, in~\cite{m.20:_weakl_poiss} software packages were developed for Mathematica, Maple and Reduce.

A minimal set of conditions that are equivalent to $[P,R]=0$ would be a big help,
but so far a solution to this problem has been out of reach.
An alternative is to figure out a suitable ansatz for Hamiltonian operators.
Known examples of bi-Hamiltonian systems of WDVV-type~\cite{vasicek21:_wdvv_hamil} allow us
to formulate the following conjecture.

\begin{conjecture}\label{conj:comp-g}
  Let a Ferapontov-type operator~$P$ parameterised by a metric~$g_{ij}$ and a
  Doyle--Pot\"emin operator~$R$ parameterised by a Monge metric~$f_{ij}$ with
  the Monge decomposition $f_{ij}=\psi^\alpha_i\phi_{\alpha\beta}\psi_j^\beta$
  be compatible.  Then the symmetric matrix
  \begin{equation*}
    Q^{\alpha\beta} = \psi^\alpha_ig^{ij}\psi^\beta_j
  \end{equation*}
  has entries that are second-degree polynomials in the field
  variables.
\end{conjecture}
In turn, it means that a suitable ansatz for $g$ is
\begin{equation*}
  g^{ij} = \psi_\alpha^i Q^{\alpha\beta} \psi_\beta^j
  \quad\text{or, in a matrix form,}\quad g=\Psi^{-1}Q(\Psi^{-1})^\top.
\end{equation*}
Therefore, whenever finding a Ferapontov-type operator with a generic ansatz is
not feasible, we resort to the above ansatz.  Luckily, when we are able to make
a general computation without resorting to an ansatz, the result is always
within the ansatz, thus indicating that the ansatz might be a feature of all
first-order operators $P$ which are compatible with a third-order operator~$R$
as above.

Finally, note that although $(Q^{\alpha\beta})$ is quadratic in its entries, it
is not a Monge metric in general.

\subsection[Classification of bi-Hamiltonian equations of WDVV-type]{
  Classification of bi-Hamiltonian equations\\ of WDVV-type}
\label{sec:class-bi-hamilt}

We classify quasilinear first-order systems of PDEs in a conservative
form~\eqref{eq:ConsHydroSys},
\[u_t=(V^i)_x=V^i_su^s_x,\]
that are bi-Hamiltonian with respect to a pair of operators $P$, $R$. Here, $P$
is a Ferapontov operator of the form
\begin{gather*}
  P^{ij}=g^{ij}\mathrm D_x + \Gamma^{ij}_s u^s_x\\
  + c^{11}V^i_su^s_x\mathrm D_x^{-1}V^j_ru^r_x
  + c^{12}\left(V^i_su^s_x\mathrm D_x^{-1}u^j_x
  + u^i_x\mathrm D_x^{-1}V^j_su^s_x\right)
  + c^{22}u^i_x\mathrm D_x^{-1}u^j_x;
\end{gather*}
Note that we used only two commuting vector fields in $P$, in view of the
discussion in Subsection~\ref{sec:first-order-hamilt}. The operator $R$ is a
third-order homogeneous local Hamiltonian operator $R$, and can always be
brought to the Doyle--Pot\"emin canonical form
by a transformation of the dependent variables.  Such a transformation neither
changes the locality (resp., nonlocality) of the Ferapontov operator~$P$, nor
it changes the shape of the system~\eqref{eq:ConsHydroSys}.

In turn, Doyle--Pot\"emin operators are classified for $n=1,\dots,4$ with
respect to several group actions~\cite{FPV14,FPV16}:
\begin{itemize}
\item the maximal group of transformations of the dependent variables that
  preserves the Doyle--Pot\"emin canonical form of~$R$ --- the group of affine
  transformations $\tilde{u}^i=A^i_ju^j + A^i_0$ with constant~$A^i_j$
  and~$A^i_0$;
\item the group of projective reciprocal transformations that fix $t$ (as it
  was found in~\cite{FPV14}),
\begin{equation}
\begin{split}
  \label{eq:29}
  \tilde{u}^i = \frac{A^i_j u^j + A^i_0}{\Delta},\quad \Delta = A^0_ju^j + A^0_0,
  \\
  \mathrm d\tilde{x} = \Delta\,\mathrm dx + (A^0_i V^i + A^0_0)\mathrm dt,
  \qquad \mathrm d\tilde{t}=\mathrm dt;
\end{split}
\end{equation}
\item the group of general projective reciprocal transformations~\cite{FPV16},
  which is generated by a $t$-fixing projective reciprocal transformation as
  above and an $x \leftrightarrow t$ inversion.
\end{itemize}

It is clear that each of the above groups is a subgroup of the next one in the
list.  We will use the above groups to classify bi-Hamiltonian systems of
WDVV-type in low dimensions.

It was proved in \cite{FPV17:_system_cl} that a first-order system of
conservation laws \eqref{eq:ConsHydroSys} admits a third-order homogeneous Hamiltonian
operator (in a canonical form~\eqref{eq:18}) if and only if
\begin{subequations}\label{eq:1}
  \begin{align}
    & f_{is}V^{s}_{j}=f_{js}V^s_{i},\\
    & c_{skj}V^s_{i}+c_{sik}V^s_{j}+c_{sji}V^s_{k}=0,\\
    &f_{ks}V^s_{ij} = c_{ksj}V^s_{i} + c_{ksi}V^s_{j}.
  \end{align}
\end{subequations}

It is important to recall that systems of the type~\eqref{eq:ConsHydroSys} that are
Hamiltonian with respect to third-order operators are linearly degenerate (or
weakly nonlinear, in another terminology) and non-diagonalisable
\cite{FPV17:_system_cl}.

The system~\eqref{eq:1} can be used in two ways.  Given a conservative system
defined by $(V^i)$ one can determine the matrix~$f$ parameterising a
third-order operator $R$.  Since the coefficients of $f$ are second-degree
polynomials, the problem becomes algebraic in nature.  On the other hand, given
the matrix~$f$ one can determine all the systems that are Hamiltonian with
respect to the corresponding operator~$R$. This problem is completely solved
in~\cite{FPV17:_system_cl}: we have $V^i = \psi^i_\gamma W^\gamma$, where
$W^\gamma=\eta^\gamma_s u^s + \xi^\gamma$, and the constants $\eta^\gamma_m$
and $\xi^\gamma$ fulfill a certain linear algebraic system.

Therefore, given a Doyle--Pot\"emin operator~$R$ it is an algebraic problem
given by the system~\eqref{eq:1} to find all quasilinear systems in a
conservative form that are Hamiltonian with respect to the operator~$R$.  We
simplify the obtained family of systems by equivalence transformations to get a
subfamily~$\mathcal S$ of nonlinear systems,
$
  u^i_t=V^i_ju^j_x.
$
In particular, every third-order homogeneous Hamiltonian
operator in dimension~one can be transformed to $\mathrm D_x^3$ using
transformations of the dependent variables only. But the quasilinear
conservative systems that are Hamiltonian with respect to $\mathrm D_x^3$ are
linear due to~\eqref{eq:1}.  Therefore, we omit altogether the dimension~one
in the classification below.

Note that also systems as above do not change their form when subject to a
general projective reciprocal transformation, see \cite{FPV17:_system_cl}.

Finally, we check if the family~$\mathcal S$ admits also a Ferapontov
operator~$P$.  It is done by direct computation, assisted by computer algebra.

We use the packages from the paper~\cite{m.20:_weakl_poiss} in order to compute
Schouten bracket $[P,R]$, with given~$R$ and unknown~$P$ (when possible).
Namely, we collect all coefficients of the Schouten bracket as a three-vector
and require them to vanish, add to this system the system of conditions under
which~$P$ is Hamiltonian.  If we can solve this system, we have an answer.

Otherwise, we assume Conjecture~\ref{conj:comp-g} to hold.  Finding a
Ferapontov operator thus reduces to the three following steps.  Firstly, we
solve the commutativity condition $g^{is}V^j_s=g^{js}V^i_s$.  Secondly, we
solve the commutativity condition $\Gamma^{sij}V^k_s = \Gamma^{skj}V^i_s$
coming from~\eqref{eq:185} and the fact that the system~$\mathcal S$ is in a
conservative form, where
$\Gamma^{lij}=\frac12(g^{is}g^{jl}_{,s} + g^{ls}g^{ji}_{,s} -
g^{js}g^{il}_{,s})$.  At this point, the matrix~$Q$ and therefore~$g$ is
completely determined and to find the matrix~$(c^{\alpha\beta})$ we solve the
system~\eqref{eq:20}.

Strictly speaking, we do not have a proof of the fact that the operator~$P$
preserves its shape under a general projective reciprocal transformation
thus remaining of Ferapontov type. However, the transformed operator will be
again homogeneous of degree $1$, and computational experiments show that it
will be again of the same type. We conjecture that this is a general property.

\section{Two-component systems}
\label{sec:cases-n=1-n=2}

\subsection{Affine classification}
\label{sec:case-n=2-affine}

It was shown in~\cite{FPV14} that up to affine transformations
$\tilde u^i=A^i_j u^j + A^i_0$, $i=1$, \dots, $n$, there are three distinct
third-order homogeneous Hamiltonian operators,
\begin{gather*}
R^{(1)}=\begin{pmatrix} 1 & 0 \\ 0 & 1 \end{pmatrix}\mathrm D_x^3,\quad
R^{(2)}=\mathrm D_x\begin{pmatrix} 0 & \mathrm D_x\frac1{u^1} \\
  \frac1{u^1}\mathrm D_x & \frac{u^2}{(u^1)^2}\mathrm D_x+\mathrm D_x\frac{u^2}{(u^1)^2}
\end{pmatrix}\mathrm D_x,
\\
R^{(3)}=\mathrm D_x\begin{pmatrix} \mathrm D_x & \mathrm D_x\frac{u^2}{u^1} \\
  \frac{u^2}{u^1}\mathrm D_x &
  \frac{(u^2)^2+1}{2(u^1)^2}\mathrm D_x+\mathrm D_x\frac{(u^2)^2+1}{2(u^1)^2}
\end{pmatrix}\mathrm D_x.
\end{gather*}

The operator~$R^{(1)}$ is admitted by linear systems of PDEs and therefore we
pay no attention to it here.

The operator~$R^{(2)}$ is admitted by a family of quasilinear systems
\begin{equation*}
  u^1_t=(\alpha u^1 + \beta u^2)_x,\quad
  u^2_t=\left(\alpha u^2+\frac{\beta(u^2)^2+\gamma}{u^1}\right)_x.
\end{equation*}
There are two inequivalent cases: $(\beta,\alpha)=(1,0)$ and
$(\beta,\alpha)=(0,1)$.  The second case is degenerate in a sense that the
corresponding system is partially coupled and one equation can be solved
explicitly, effectively making the second one linear. So, we will focus on the
first case only.

\begin{theorem}\label{th:case2}
The system
\[
 u^1_t=u^2_x,\quad
 u^2_t=\left(\frac{(u^2)^2+\gamma}{u^1}\right)_x
\]
is Hamiltonian with respect to three first-order local homogeneous Hamiltonian
operators $P^{(2,i)}$, $i=1$, $2$, $3$, parameterised by the metrics
\begin{gather*}
  g^{(2,1)}=\begin{pmatrix} - u^1 & 0 \\
    0 &  \frac{(u^2)^2+\gamma}{u^1} \end{pmatrix},
  \quad
  g^{(2,2)}=\begin{pmatrix} 0 & u^1
    \\ u^1 & 2u^2 \end{pmatrix},
  \\
  g^{(2,3)}=\begin{pmatrix} 2u^2 & \frac{(u^2)^2+\gamma}{u^1} \\
    \frac{(u^2)^2+\gamma}{u^1} & 0 \end{pmatrix}.
\end{gather*}
The first-order operators above are mutually compatible as well as compatible
with the operator~$R^{(2)}$.
\end{theorem}
The operators $P^{(2,i)}$ belong to the list~\cite{LSV:bi_hamil_kdv}
of all first-order local homogeneous Hamiltonian operators that are compatible
with~$R^{(2)}$.

The operator~$R^{(3)}$ is admitted by a family of quasilinear systems
\begin{equation*}
  u^1_t=(\alpha u^1 + \beta u^2)_x,\quad
  u^2_t=\left(\alpha u^2+\frac{\beta (u^2)^2+\gamma u^2 - \beta}{u^1}\right)_x.
\end{equation*}
Again, there are two inequivalent cases: $(\beta,\alpha)=(1,0)$ and
$(\beta,\alpha)=(0,1)$, and we will only deal with the first one.
\begin{theorem}\label{th:case3}
The system
\[
 u^1_t=u^2_x,\quad
  u^2_t=\left(\frac{(u^2)^2+\gamma u^2 - 1}{u^1}\right)_x
\]
is Hamiltonian with respect to three first-order local homogeneous Hamiltonian
operators $P^{(3,i)}$, $i=1,2,3$, parameterised by the metrics
\begin{gather*}
  g^{(3,1)}=\begin{pmatrix} -u^1 & 0 \\
    0 & \frac{(u^2)^2+\gamma u^2-1}{u^1} \end{pmatrix},
  \quad
  g^{(3,2)}=\begin{pmatrix} 0 & u^1
    \\ u^1 & 2u^2+\gamma \end{pmatrix},
  \\
  g^{(3,3)}=\begin{pmatrix} 2u^2+\gamma &  \frac{(u^2)^2+\gamma u^2 - 1}{u^1} \\
    \frac{(u^2)^2+\gamma u^2 - 1}{u^1} & 0 \end{pmatrix}.
\end{gather*}
The first-order operators above are mutually compatible as well as compatible
with the operator~$R^{(3)}$.
\end{theorem}
Again, the operators $P^{(3,i)}$ belong to the list~\cite{LSV:bi_hamil_kdv} of
all first-order local homogeneous Hamiltonian operators that are compatible
with $R^{(3)}$. To the best of our knowledge, the multi-Hamiltonian systems in
Theorems~\ref{th:case2} and~\ref{th:case3} are not known in the literature.

\begin{proposition}
  Each of the metrics $(g^{ij})_{i,j=1}^2$ of first-order local
  homogeneous Hamiltonian operators $P$ that are compatible with a third-order
  local homogeneous Hamiltonian operator $R$ (with Monge metric
  $f_{ij}=\phi_{\alpha\beta}\psi^\alpha_i\psi^\beta_j$), listed in
  \cite{LSV:bi_hamil_kdv}, can be factorised as in
  Conjecture~\ref{conj:comp-g}:
  \begin{equation*}
    \label{eq:11}
    g^{ij}=\psi^i_\alpha Q^{\alpha\beta}\psi^j_\beta,
  \end{equation*}
  where $Q=(Q^{\alpha\beta})$ is a symmetric matrix whose entries are quadratic
  polynomials of field variables.
\end{proposition}

As an example, the Monge metric $f^{(3)}$ of the operator~$R^{(3)}$ admits the
decomposition~\eqref{eq:12} with
\begin{equation*}
h_3=\begin{pmatrix} (u^2)^2+1 & -u^1u^2 \\ -u^1u^2 & (u^1)^2 \end{pmatrix},\quad
\Psi=\begin{pmatrix} -u^2 & 1 \\ u^1 & 0 \end{pmatrix},\quad
\Phi=\begin{pmatrix} 1 & 0 \\ 0 & 1 \end{pmatrix},
\end{equation*}
while the Hamiltonian operator of Dubrovin--Novikov type parameterised by the
metric $g:=c_1g^{(3,1)}+c_2g^{(3,2)}+c_3g^{(3,3)}$ admits the decomposition~\eqref{eq:11} with
\begin{multline*}
Q =
  \left(
\begin{matrix}
\gamma(c_2{-}c_3) (u^2)^2{+}(c_1\gamma u^2 {-}c_1)u^1{+}2c_3u^2
&
c_2(u^1)^2{+}c_1u^1u^2{-}c_3 (u^2)^2{-}c_3\\
c_2(u^1)^2{+}c_1u^1u^2{-}c_3 (u^2)^2{-}c_3
&
-c_1u^1+2c_3 u^2+c_3\gamma
\end{matrix}
\right).
\end{multline*}

\subsection{Projective classification}
\label{sec:case-n=3-projective}

It follows from the works of Agafonov and Ferapontov
\cite{agafonov96:_system,agafonov99:_theor,agafonov01:_system_templ} that the
quasilinear first-order systems that are Hamiltonian with respect to a
third-order homogeneous Hamiltonian operator~\cite{FPV17:_system_cl} correspond
to \emph{linear line congruences}. These are algebraic varieties in the
Pl\"ucker embedding of the Grassmannian of lines in $\mathbb{P}^{n+1}$.

The classical works of Castelnuovo imply that all such congruences can be
transformed to a single one (see~\cite{FPV17:_system_cl}) when $n=2$, which
implies that both operators~$R^{(2)}$ and~$R^{(3)}$ can be transformed to
$R^{(1)}$ by a reciprocal projective transformation.  Thus, there is no
interesting case in dimension~$2$ under the action of the above group.

More precisely, we have the following result.
\begin{proposition}
  Any two-component quasilinear system in a conservative form is linearisable
  by a reciprocal projective transformation.
\end{proposition}

It can be very difficult to find the above linearizing reciprocal projective
transformation explicitly, even in this low-dimensional case.
Alternatively, as a canonical case of the classification one can take a physically relevant system.
Thus, the Chaplygin gas system
\[
u_t+uu_x+\frac{v_x}{v^3}=0,\quad
v_t+(uv)_x=0,
\]
is known~\cite{mokhov98:_sympl_poiss} to admit three first-order Dubrovin--Novikov Hamiltonian operators.
Moreover, as a subcase of a polytropic gas dynamics system,
\[
u_t+uu_x+v^\gamma v_x=0,\quad
v_t+(uv)_x=0,
\]
it admits~\cite{OlverNutku1988} a nonhomogeneous third-order Hamiltonian operator.
But it also admits a homogeneous one. Indeed, the diagonalised form of the Chaplygin gas system~\cite{mokhov98:_sympl_poiss} is
\[
U_t=VU_x,\quad V_t=UV_x,\quad\left(\text{with}\quad U=u-\frac1v,\quad V=u+\frac1v\right)
\]
to which the system
\[
 u^1_t=u^2_x,\quad
 u^2_t=\left(\frac{(u^2)^2-a^2}{u^1}\right)_x
\]
is reduced with the help of the point transformation of the dependent variables
\[
u^1 = \frac{2a}{U-V},\quad u^2 = \frac{a(U+V)}{U-V},
\]
and the system
\[
 u^1_t=u^2_x,\quad
  u^2_t=\left(\frac{(u^2)^2+\gamma u^2 - 1}{u^1}\right)_x
\]
is reduced with the help of the point transformation of the dependent variables
\[
u^1 = -\frac{\sqrt{\gamma^2+4}}{U-V}, \quad
u^2 = -\frac{\sqrt{\gamma^2+4}}2\frac{U+V}{U-V}-\frac\gamma2.
\]
Since both pairs of transformations preserve the homogeneity of Hamiltonian
operators, the Chaplygin gas system possesses a homogeneous third-order
Hamiltonian operator (which is not in Doyle--Pot\"emin canonical form).

\section{Three-component systems}
\label{sec:case-n=3-projective-1}

An affine classification is no longer feasible when $n>2$, due to a large
amount of cases and subcases that we would have. We resort to two distinct
projective classifications.

\subsection{Partial projective classification}
\label{sec:proj-class-with}

Under the action of projective reciprocal transformations that fix~$t$, there
are six classes of third-order homogeneous Hamiltonian operators $R^{(i)}$,
see~\cite{FPV14}. They are defined by the following Monge metrics,
\begin{gather*}
  f^{(1)}=\begin{pmatrix} (u^{2})^{2}+\mu & -u^{1}u^{2}-u^{3} & 2u^{2} \\
    -u^{1}u^{2}-u^{3} & (u^{1})^{2}+\mu(u^{3})^{2} & -\mu u^{2}u^{3}-u^{1} \\
    2u^{2} & -\mu u^{2}u^{3}-u^{1} & \mu(u^{2})^{2}+1
  \end{pmatrix},
  \\[1ex]
   f^{(2)} = \begin{pmatrix}
    (u^{2})^{2}+1 & -u^{1}u^{2}-u^{3} & 2u^{2} \\
    -u^{1}u^{2}-u^{3} & (u^{1})^{2} & -u^{1} \\
    2u^{2} & -u^{1} & 1
  \end{pmatrix},
  \
   f^{(3)} = \begin{pmatrix}
    (u^{2})^{2}+1 &  -u^{1}u^{2}&0 \\
    -u^1u^2 & (u^1)^2 & 0 \\
    0 & 0 & 1%
  \end{pmatrix},
\\[1ex]
   f^{(4)}= \begin{pmatrix} -2u^2 & u^1 & 0
    \\
    u^1 & 0 & 0
    \\
    0 & 0 & 1
  \end{pmatrix},
  \quad
   f^{(5)}=\begin{pmatrix} -2u^2 & u^1 & 1
    \\
    u^1 & 1 &0
    \\
    1 & 0 & 0
  \end{pmatrix},
  \quad
   f^{(6)} =
  \begin{pmatrix}
    1 & 0 & 0\\ 0 & 1 & 0\\ 0 & 0 & 1
  \end{pmatrix}.
\end{gather*}
The corresponding quasilinear first-order systems of evolutionary PDEs have
been found and described in \cite{FPV16,FPV17:_system_cl}.

\paragraph{Case $R^{(6)}$.} The system which is Hamiltonian with respect to~$R^{(6)}$ is linear, hence it is out of consideration here.

\paragraph{Case $R^{(5)}$.}  The conservative quasilinear system of PDEs that
is determined by the operator~$R^{(5)}$ is~\eqref{eq:27}, which is equivalent
to a WDVV equation. The bi-Hamiltonian pair for such a system was found
in~\cite{FGMN97}, see Introduction.

\paragraph{Case $R^{(4)}$.} The system of PDEs that is determined by~$R^{(4)}$
is
  \begin{equation*}
    u^1_t=u^2_x,\quad u^2_t=\left( \frac{(u^2)^2+u^3}{u^1}\right)_x,
    \quad u^3_t=u^1_x,
\end{equation*}
Setting $u^1=f_{xxt}$, $u^2=f_{xtt}$, $u^3=f_{xxx}$ we obtain
$f_{xxx} =f_{ttt}f_{xxt}- f_{xtt}^2$, which is equivalent to the WDVV
equation~\eqref{eq:26} under the interchange of~$x$ and~$t$.  Its
bi-Hamiltonian representation by means of a compatible pair~$P$ and~$R^{(4)}$
as in this paper was constructed in~\cite{kalayci97:_bi_hamil_wdvv,kalayci98:_alter_hamil_wdvv}.

\paragraph{Case $R^{(3)}$.} The integrability of the system of PDEs determined
by~$R^{(3)}$,
  \begin{equation*}
    u^1_t=(u^2+u^3)_x, \quad u^2_t=\left(\frac{u^2(u^2+u^3)-1}{u^1}\right)_x,
    \quad
    u^3_t=u^1_x,
\end{equation*}
was first determined in~\cite{agafonov98:_linear}. Its third-order Hamiltonian
structure was found
in~\cite{FPV17:_system_cl}. In~\cite{vasicek21:_wdvv_hamil}, a criterion
from~\cite{bogoyavlenskij96:_neces_hamil} was used to find the metric~$g$ of a
first-order operator. Namely, for non-diagonalisable quasilinear systems of
PDEs, the metric~$g$ is proportional to a contraction of the square of the
Haantjies tensor of the velocity matrix $(V^i_j)$ of the system:
\begin{equation*}
  g_{ij} = f(u) H^\alpha_{i\beta}H^{\beta}_{j\alpha}
\end{equation*}
(see~\cite{bogoyavlenskij96:_neces_hamil} for the definition of the Haantjies tensor).

What was new is that in this case the first-order homogeneous Hamiltonian
operator $P$ was nonlocal, $c^{11}=c^{22}=-1$, $c^{12}=c^{21}=0$.  The metric
is written down in~\cite{m.20:_weakl_poiss}, and here we present it in accordance
with Conjecture~\ref{conj:comp-g}.  The Monge metric $f^{(3)}$ admits a Monge
decomposition $(f^{(3)}_{ij})=\Psi\Phi\Psi^{\top}$, where the metric
$g=\Psi^{-1} Q(\Psi^{-1})^\top$, where
\begin{gather*}
Q^{11}=4(u^1)^2+(u^2)^2+1,\quad
Q^{12}=-3u^1,\quad
Q^{13}=-2u^2-u^3,\\
Q^{22}=(u^1)^2+(u^3)^2+4,\quad
Q^{23}=u^1(u^2+2u^3),\\
Q^{33}=(u^1)^2+(u^2+2u^3)^2+1,\\
\Psi=(\psi^\alpha_i) = \begin{pmatrix}
-u^2 & 0 & 1 \\
u^1 & 0 & 0 \\
0 & 1 & 0
\end{pmatrix},\quad
\Phi=(\phi_{\alpha\beta})=\begin{pmatrix}
1 & 0 & 0 \\
0 & 1 & 0\\
0 & 0 & 1
\end{pmatrix}.
\end{gather*}

The system of PDEs that is determined by $R^{(3)}$ is equivalent to a
particular WDVV equation~\cite{D96} obtained when $\eta=\text{Id}$, the
identity matrix.  Setting $u^1=f_{xxt}$, $u^2=f_{xtt}-f_{xxx}$, $u^3=f_{xxx}$,
we have
\[f_{xxt}^2 -f_{xxx}f_{xtt}+ f_{xtt}^2-f_{xxt}f_{ttt}-1=0.\]

At this point, it is natural to ask the question if in the case $n=3$ all the
items in the projective classification of operators~$R$ determine WDVV systems
of PDEs.  The answer is negative.
Indeed, WDVV equations in dimension~three were classified
in~\cite{mokhov18:_class_hamil}, and in~\cite{vasicek21:_wdvv_hamil} the
bi-Hamiltonian pairs for each member of the classification were found to be the
pairs considered in this paper.  It turned out that the projective classes of
the third-order homogeneous Hamiltonian operators $R$ were that of~$R^{(3)}$,
$R^{(4)}$, $R^{(5)}$. It means that the systems of PDEs that are determined
by~$R^{(1)}$ and~$R^{(2)}$ \emph{are not} WDVV systems.  That is one of the
main motivations for this paper: there are more WDVV-type equations than WDVV
systems.

\paragraph{Case $R^{(2)}$.}
A Monge decomposition $(f^{(2)}_{ij})=\Psi\Phi\Psi^{\top}$ of the Monge
metric~$f^{(2)}$ is given by
\[
\Psi=(\psi^\alpha_i) = \begin{pmatrix}
u^2 & 0 & 1 \\
-u^1 & -u^3 & 0 \\
1 & u^2 & 0
\end{pmatrix},\quad
\Phi=(\phi_{\alpha\beta})=\begin{pmatrix}
1 & 0 & 0 \\
0 & 0 & 1\\
0 & 1 & 1
\end{pmatrix}.
\]
The system $u^i_t=(V^i)_x$ admitted by $R^{(2)}$ is (up to a transformation in
the stabilizer of $R^{(2)}$ that reduces the number of constants)
\begin{equation}\label{eq:35}
\begin{split}
  u^1_t =& (\alpha u^2+\beta u^3)_x,
  \\
  u^2_t=& \left(\frac{((u^2)^2-1)(\alpha u^2+\beta u^3)-(\gamma +\delta  u^1)}{S}\right)_x,
  \\
  u^3_t= &\left(\frac{(u^2u^3-u^1)(\alpha u^2+\beta u^3)-u^1(\gamma +\delta u^1)}{S}\right)_x,
\end{split}
\end{equation}
where $S= u^1u^2-u^3$ is proportional to $\sqrt{\det(f^{(2)})}$ and
$\alpha, \beta, \gamma, \delta$ are arbitrary constants. The system admits a
Hamiltonian formulation through $R^{(2)}$ and the nonlocal Hamiltonian
\[
 H= \int \left(\frac\alpha2 u^3 (\mathrm D_x^{-1}u^2)^{2}
      +  \beta u^3 (\mathrm D_x^{-1}u^2)  (\mathrm D_x^{-1}u^3)
      -\frac\gamma{2} x^{2}u^1
      -  \delta x u^1 (\mathrm D_x^{-1}u^1)\right) \mathrm dx.
\]
Again, this system is linearly degenerate, and non-diagonalisable for generic
values of parameters; it is diagonalisable if and only if
$\alpha \delta - \beta \gamma=0$.

\begin{theorem}\label{th:case-r2}
  The system~\eqref{eq:35} has a unique Ferapontov-type Hamiltonian
  operator~$P$, which is compatible with~$R^{(2)}$, fulfills Conjecture~\ref{conj:comp-g} with the metric
  $(g^{ij})=\Psi^{-1} Q(\Psi^{-1})^\top$, $c^{11}=3$, $c^{12}=c^{21}=0$,
  $c^{22}=-\beta^2$.  Here, 
\begin{gather*}
Q^{11}=2(A^2+B^2+4BC+2AC),\quad
Q^{12}=2(3AD-BC),\\
Q^{13}=2B(2A+3C),\quad
Q^{22}=-2(2A+C)(2A+3C ),\\
Q^{23}=8A^2+10AC+2BD,\quad
Q^{33}=-6A^2+2B^2,\\
A=\alpha u^2+\beta u^3,\quad B=\beta u^1+\alpha,\quad C=\delta u^1+\gamma,\quad D=\delta u^3+\gamma u^2.
\end{gather*}
\end{theorem}

\paragraph{Case $R^{(1)}$.} In this subsection, $\mu^2\neq1$.
The Monge metric $f^{(1)}$ admits a Monge decomposition
$(f^{(1)}_{ij}) =\Psi\Phi\Psi^{\top}$, where
\[
(\psi^\alpha_i)=\begin{pmatrix}
u^2 & 0 & 1 \\
-u^1 & -u^3 & 0 \\
1 & u^2 & 0
\end{pmatrix},\quad
(\phi_{\alpha\beta})=\begin{pmatrix}
1 & 0 & 0 \\
0 & \mu & 1\\
0 & 1 & \mu
\end{pmatrix}.
\]
The system $u^i_t=(V^i)_x$ admitted by $R^{(1)}$ is (up to a transformation
in the stabilizer of $R^{(2)}$ that reduces the number of constants) is
\begin{equation}\label{eq:32}
  \begin{split}
    u^1_t = &(\alpha u^2+\beta u^3)_x,
    \\
    u^2_t = &\left(\frac{\left((u^2)^2-\mu\right)(\alpha u^2+\beta u^3) +\gamma (1-\mu(u^2)^2) +\delta(u^1-\mu u^2u^3)}{S}\right)_x,
    \\
    u^3_t= & \Big(\frac{\alpha u^3\left((u^2)^2-\mu\right)+\beta u^3(u^2u^3-\mu u^1)}{S}\\
    & \qquad+\frac{\gamma (u^1-\mu u^2u^3)+\delta \left((u^1)^2-\mu(u^3)^2\right)}{S}\Big)_x,
\end{split}
\end{equation}
where $S=u^1u^2-u^3$ is proportional to $\sqrt{\det(f^{(1)})}$, and
$\alpha, \beta, \gamma, \delta$ are arbitrary constants. Its
nonlocal Hamiltonian is
 $$
\begin{array}{c} H= {\displaystyle\int} \Big(
 \frac\alpha2
      (2\mu xu^1 \mathrm D_x^{-1}  u^{2}
        +u^3(\mathrm D_x^{-1}u^{2})^{2}
        + \mu    x^{2}u^3
      )
      +  \beta  u^3(1-\mu ^2) \mathrm D_x^{-1}u^{2}  \mathrm D_x^{-1}u^{3}
        \\
        \ \\
      +  \delta(
      xu^1\mathrm D_x^{-1}u^{1}
        +\mu u^3\mathrm D_x^{-1}u^{1}  \mathrm D_x^{-1}u^{2}
        +\mu u^1\mathrm D_x^{-1} u^{2}  \mathrm D_x^{-1}u^{3}
        +\mu xu^{3}  \mathrm D_x^{-1}u^3 )
        \\
        \ \\
      +\frac{1}{2}\gamma
      (\mu u^1(\mathrm D_x^{-1}u^{2})^{2}
        +x^2u^1
        +2\mu xu^3\mathrm D_x^{-1}  u^{2}
\Big) \mathrm dx.
      \end{array}
 $$
The system~\eqref{eq:32} is linearly degenerate, and non-diagonalisable
for generic values of parameters ($\alpha \delta - \beta \gamma\neq0$).

\begin{theorem}
  The system~\eqref{eq:32} has a unique first-order
  nonlocal Hamiltonian operator~$P$ that fulfills Conjecture~\ref{conj:comp-g}
  with the metric~$(g^{ij})=\Psi^{-1} Q(\Psi^{-1})^\top$ and is compatible with~$R^{(1)}$,
\begin{gather*}
c^{11}= \mu^2+3,\quad c^{12}=c^{21}=-4\mu \delta,\quad c^{22}= \mu^3\beta^2 + 4\mu^2\delta^2 -\mu\beta^2,\\
Q^{11}=-(\mu^2-1)\left(\mu^2(A+C)^2+\mu (B^2+D^2)-2BD-4EF\right),
\\
Q^{12}=-(\mu^2-1)(\mu ED-FB),
\quad
Q^{13}=-(\mu^2-1)(\mu B(2E+F)-3DE),
\\
Q^{22}=-F^2\mu^3-\mu^2(4A^2+D^2)+\mu (8BD+F^2)-3D^2,
\\
Q^{23}=-\mu^2(2BD+(u^1u^2-u^3)(\alpha\delta-\beta\gamma))
+4\mu (B^2+D^2)-5BD-EF,
\\
Q^{33}=-\mu^3E^2-\mu^2(B^2+4D^2)+\mu (E^2+8BD)-3B^2,\\
A=\beta u^1+\delta u^3,\quad
B=\alpha u^2+\beta u^3,\quad
C=\gamma u^2+\alpha,\\
D=\delta u^1+\gamma,\quad
E=\beta u^1+\alpha,\quad
F=\delta u^3+\gamma u^2.
\end{gather*}

\end{theorem}

\subsection{General projective classification}

In~\cite{FPV17:_system_cl}, it was proved that third-order homogeneous
Hamiltonian operators in the Doyle--Pot\"emin canonical form are invariant also
with respect to transformations that exchange $t$ and $x$. This, together with
projective reciprocal transformations that fix $t$, generate a larger group of
reciprocal transformations of the following type:
\begin{gather*}
  \mathrm d\tilde{x} = (A_i u^i + A_0)\mathrm dx + (A_iV^i + C_0)\mathrm dt,
  \\
  \mathrm d\tilde{t} = (B_i u^i + B_0)\mathrm dx + (B_iV^i + D_0)\mathrm dt,
\end{gather*}
coupled with affine transformations of the dependent variables
\cite{agafonov96:_system,agafonov01:_system_templ}. Such transformations act as
$\mathrm{SL}(n+2)$ transformations on the linear line congruence that
corresponds to the quasilinear first-order systems of PDEs determined by
third-order homogeneous Hamiltonian operators.

According to classical results by Castelnuovo, there are four distinct classes
of linear line congruences in $\mathbb{P}^4$ under the action of
$\mathrm{SL}(5)$: only two of them are endowed with third-order homogeneous
Hamiltonian operator, they correspond to $R^{(5)}$ and $R^{(6)}$,
see~\cite{FPV17:_system_cl}.

It turns out that the $6$ classes that we discussed in the previous section can
be transformed to $2$ classes, one of which contains a linear system, and the
other contains the simplest WDVV equation \eqref{eq:27}. It provides a proof
that all systems that we discussed in the previous section are bi-Hamiltonian
and integrable.

However, we stress that finding equivalence transformations between a given
system or a Hamiltonian operator and a representative of an corresponding
equivalence class can be extremely challenging, even with the most advanced
computer algebra systems.  So, having a direct proof and methods to efficiently
compute systems and Hamiltonian operators proves to be an invaluable set of
tools.

\section{Four-component systems}

There exists~\cite{FPV16} a projective classification of third-order
homogeneous Hamiltonian operators in $n=4$. The group acting on the operators
is that of $t$-fixing projective reciprocal transformations.

Unfortunately, unlike cases with $n\leq 3$, not all classes of
third-order homogeneous Hamiltonian operators and associated systems admit a
compatible first-order local or nonlocal Hamiltonian operator that
fulfill the Conjecture \ref{conj:comp-g}. As a distinguished example, let us
consider systems of the type
\begin{equation}
  \label{eq:42}
 u^1_t = u^2_x,\quad u^2_t=u^3_x,\quad u^3_t = u^4_x,\quad u^4_t=(f(u))_x.
\end{equation}
Linearly degenerate systems of the above type have been studied
in~\cite{agafonov98:_linear}.  In~\cite{FPV16} it is proved that the above
system is Hamiltonian with respect to a third-order homogeneous Hamiltonian
operator $R$ only for two values of $f$:
\begin{equation}
  \label{eq:43}
  f_1(u) = (u^2)^2 - u^1u^3,\qquad f_2(u) = (u^3)^2 - u^2u^4 + u^1.
\end{equation}

\begin{proposition}
  There does not exist a matrix $(g^{ij})$ fulfilling the
  Conjecture~\ref{conj:comp-g} for the systems~\eqref{eq:42}, \eqref{eq:43}.
\end{proposition}
The above proposition does not exclude the possibility that an operator~$P$
that does not fulfill the Conjecture and is still compatible with the
operator~$R$ exists for the above system; however, given the fact that the
Conjecture has been verified in a substantial number of cases (and never be
disproven at the moment) we think that such a possibility has a little chance.

On the other hand, it is believed that there is a unique integrable case within
the class of systems of conservation laws that admits a Hamiltonian formulation
through a third-order homogeneous Hamiltonian operator (see the discussion at
the end of \cite{ferapontov18:_system_hamil}). This is represented by the
system
\begin{equation}\label{eq:44}
  \begin{split}
    &u^1_t=u^3_x,\\
&u^2_t=u^4_x,\\
&u^3_t=\left(\frac{u^1u^2u^4+u^3((u^3)^2+(u^4)^2-(u^2)^2-1)}{u^1u^3+u^2u^4}
\right)_x,\\
&u^4_t=\left(\frac{u^1u^2u^3+u^4((u^3)^2+(u^4)^2-(u^1)^2-1)}{u^1u^3+u^2u^4}
\right)_x,
\end{split}
\end{equation}
which is known to possess a Lax pair and a Doyle--Pot\"emin Hamiltonian
operator~$R$ parameterised by a Monge metric
$f=(f_{ij})$~\cite{FPV17:_system_cl},
\begin{gather*}
(f_{ij})=\begin{pmatrix}
(u^2)^2{+}(u^3)^2+1 & -u^1u^2+u^3u^4 & -u^1u^3+u^2u^4 & - 2u^2u^3 \\
-u^1u^2+u^3u^4 & (u^1)^2{+}(u^4)^2{+}1 & -2u^1u^4 & u^1u^3-u^2u^4 \\
-u^1u^3+u^2u^4 & -2u^1u^4 & (u^1)^2{+}(u^4)^2 & u^1u^2-u^3u^4 \\
- 2u^2u^3 & u^1u^3-u^2u^4 & u^1u^2-u^3u^4 & (u^2)^2{+}(u^3)^2
\end{pmatrix}
\end{gather*}
The above Monge metric can be factorised as
$f=\Psi\Phi\Psi^{\top}=\psi^\alpha_i\phi_{\alpha\beta}\psi^\beta_j$,
\begin{equation*}
\Psi=\begin{pmatrix}
-u^2 & -u^3& 1 & 0 \\
u^1 & -u^4 & 0 & 1 \\
-u^4 & u^1 & 0 & 0 \\
u^3 & u^2 & 0 & 0
\end{pmatrix},\quad
\Phi=\begin{pmatrix}
1 & 0 & 0 & 0 \\
0 & 1 & 0 & 0 \\
0 & 0 & 1 & 0 \\
0 & 0 & 0 & 1
\end{pmatrix}
\end{equation*}

It is an original result the fact that the above system is bi-Hamiltonian of
the type under consideration.
\begin{theorem}
  The system~\eqref{eq:44} is Hamiltonian with respect to a first-order
  nonlocal Hamiltonian operator~$P$ that is compatible with~$R$ and is defined
  by the metric $g=(g^{ij})$ splitting as in Conjecture~\ref{conj:comp-g},
  $c^{11}=c^{22}=1$, $c^{12}=c^{21}=0$.  Here,
\begin{gather*}
Q^{11}=(u^1)^2+(u^2)^2+(u^3)^2+(u^4)^2,\quad
Q^{12}= -2u^1u^4+2u^2u^3,\quad
Q^{13}=-u^2,\\
Q^{14}= u^1,\quad
Q^{22}= (u^1)^2+(u^2)^2+(u^3)^2+(u^4)^2+1,\quad
Q^{23}= -2u^3,\\
Q^{24}= -2u^4,\quad
Q^{33}=(u^1)^2+(u^3)^2+1,\quad
Q^{34}= u^1u^2+u^3u^4,\\
Q^{44}= (u^2)^2+(u^4)^2+1.
\end{gather*}
\end{theorem}

\section{And beyond}\label{sec:AndBeyound}

To the best of our knowledge, there is only one known example of a system with
a WDVV-type bi-Hamiltonian structure for $n>4$.  It comes from the WDVV
equations in $N=4$ dimensions~\cite{PV15} (here, $n=6$).  A new example not
related with WDVV equations is given below.

In~\cite[Eq.~(11)]{FPX2021} the problem of finding integrable Lagrangians
within a certain class is reformulated as the problem of finding solutions to
two \emph{commuting} quasilinear systems of first-order PDEs of the type
$u^i_t=V^i_j(u^k)u^j_x$.

Here, we rewrite the above two systems in conservative form; one of
them is:
\begin{equation}\label{eq:451}
  \begin{split}
    &u^1_t=\left(\frac{4 (u^1)^3 u^3-4 (u^1)^2 u^6+2 u^1 u^2 u^4-(u^4)^2}{2(2 u^1 u^2 u^3-u^1 u^5-u^3 u^4)}\right)_x,\\
    &u^2_t=u^1_x,\\
    &u^3_t=\left(\frac{4 (u^1)^2 (u^3)^2-4 u^1 u^3 u^6+u^4 u^5}{2(2 u^1 u^2 u^3-u^1 u^5-u^3 u^4)}\right)_x,\\
    &u^4_t=\left(\frac{2 (u^1)^3 u^5-4 (u^1)^2 u^2 u^6+2 (u^1)^2 u^3 u^4+2 u^1 (u^2)^2 u^4-u^2 (u^4)^2}{2(2 u^1 u^2 u^3-u^1 u^5-u^3 u^4)}\right)_x,\\
    &u^5_t=\left(\frac{2 (u^1)^2 u^3 u^5+2 u^1 (u^3)^2 u^4-2 u^1 u^5 u^6+u^2 u^4 u^5-2 u^3 u^4 u^6}{2(2 u^1 u^2 u^3-u^1 u^5-u^3 u^4)}\right)_x,\\
    &u^6_t=\frac12 u^4_x.
\end{split}
\end{equation}

\begin{proposition}
  System~\eqref{eq:451} possesses a third-order Hamiltonian operator~$R$
  parameterised by a Monge metric $f=(f_{ij})$, where
  \begin{gather*}
(f_{ij})=\begin{pmatrix}
(u^3)^2                       & -\frac{u^5+u^2u^3}2      & -u^1u^3+\frac{(u^2)^2}2    & 0           & -\frac{u^2}2 & 0 \\
-\frac{u^5+u^2u^3}2          & u^1u^3+u^6                & -\frac{u^4+u^1u^2}2      & u^3         & u^1         & -\frac{u^2}2 \\
-u^1u^3+\frac{(u^2)^2}2        & -\frac{u^4+u^1u^2}2      & (u^1)^2                   & -\frac{u^2}2 & 0           & 0 \\
0                             & u^3                       & -\frac{u^2}2               & 0           & \frac12     & 0 \\
-\frac{u^2}2                   & u^1                       & 0                         & \frac12     & 0           & 0 \\
0                             & -\frac{u^2}2               & 0                         & 0           & 0           & 1
\end{pmatrix}
\end{gather*}
\end{proposition}
\begin{proof}
  We make use of equations~\eqref{eq:1} for the above system, with an unknown Monge metric.
  We find the above Monge metric as the unique solution (up to constant multiples).
\end{proof}

The above Monge metric can be factorised as
$f=\Psi\Phi\Psi^{\top}=\psi^\alpha_i\phi_{\alpha\beta}\psi^\beta_j$ where
\begin{equation*}
\Psi=\begin{pmatrix}
-u^3 & -u^2 & u^5 & 0 & 0 &0\\
0 & u^1 & -u^6 & u^3 & 1 &0\\
u^1 & 0 & u^4 & -u^2 & 0 &0\\
0 & 1 & -u^3 & 0 & 0 &0\\
0 & 0 & -u^1 & 1 & 0 &0\\
0 & 0 & u^2 & 0 & 0 &1
\end{pmatrix},\quad
\Phi=\begin{pmatrix}
1 & 0 & 0 & 0 & 0 &0 \\
0 & 0 & 0 & \frac12 & 0 &0 \\
0 & 0 & 0 & 0 & -\frac12 &0 \\
0 & \frac12 & 0 & 0 & 0 &0 \\
0 & 0 & -\frac12 & 0 & 0 &0 \\
0 & 0 & 0 & 0 & 0 &1
\end{pmatrix}
\end{equation*}

\begin{theorem}
  The system~\eqref{eq:451} is Hamiltonian with respect to a first-order local
  Hamiltonian operator $P$ that is defined by the metric $g=(g^{ij})$ splitting
  as in Conjecture~\ref{conj:comp-g} and compatible with~$R$.  Here,
\begin{gather*}
Q^{11}=2u^1u^3,\quad
Q^{12}= u^1u^2-u^4,\quad
Q^{13}=-2u^1u^5+2u^3u^4,\\
Q^{14}= -u^2u^3+u^5,\quad
Q^{15}= 0,\quad
Q^{16}= 0,\quad
Q^{22}= 2(u^1)^2,\\
Q^{23}= -4u^1u^6+2u^2u^4,\quad
Q^{24}= 2u^1u^3-(u^2)^2+4u^6,\quad
Q^{25}= 4u^1,\\
Q^{26}= 2u^4,\quad
Q^{33}=-2u^4u^5+2(u^6)^2,\quad
Q^{34}= 2u^2u^5-4u^3u^6,\\
Q^{35}= -4u^1u^3+(u^2)^2-2u^6,\quad
Q^{36}= u^2u^6-u^1u^5-u^3u^4,\quad
Q^{44}= 2(u^3)^2,\\
Q^{45}= 4u^3,\quad
Q^{46}= 2u^5,\quad
Q^{55}= 2,\quad
Q^{56}= u^2,\quad
Q^{24}= 2u^6.
\end{gather*}
\end{theorem}

\begin{remark}
  It is interesting to observe that there is another quasilinear first-order
  system in a conservative form that can be deduced from the two quasilinear
  systems in~\cite[Eq. (11)]{FPX2021}. The above bi-Hamiltonian structure of
  WDVV-type holds also for this system, a phenomenon that has already been
  observed for the four-component WDVV systems in~\cite{ferapontov96:_hamil}.
\end{remark}

\begin{remark}
  A further example of bi-Hamiltonian structure of WDVV-type is incomplete, but
  interesting. Indeed, the simplest case of Oriented Associativity equation has
  a first-order local operator~$P$ and a third-order \emph{nonlocal}
  operator~$R$, which suggests the possibility that the class of bi-Hamiltonian
  structures of WDVV type can be further enlarged.

  It must be stressed that in the above case compatibility of the two operators
  $P$, $R$ has never been proved, also in view of the computational complexity
  of the problem.
\end{remark}

\section*{Acknowledgements}

We wish to express our thanks to A. Caruso, E.V. Ferapontov, P. Lorenzoni and
A.C. Norman for stimulating discussions.

SO and RV acknowledge financial support from Istituto Nazionale di Fisica
Nucleare (INFN) by IS-CSN4 \emph{Mathematical Methods of Nonlinear Physics}
(MMNLP), from the GNFM of Istituto Nazionale di Alta Matematica ``F. Severi''
(INdAM) and from Dipartimento di Matematica e Fisica ``E. De Giorgi'' of the
Universit\`a del Salento. RV also acknowledges financial support from PRIN 2022
(2022TEB52W) - PE1 - \emph{The charm of integrability: from nonlinear waves to
  random matrices}.

This work is (partially) supported by ICSC -- Centro
Nazionale di Ricerca in High Performance Computing, Big Data and Quantum
Computing, funded by European Union -- NextGenerationEU.

\bibliographystyle{hplain}

\providecommand{\cprime}{\/{\mathsurround=0pt$'$}}
  \providecommand*{\SortNoop}[1]{}

\end{document}